\documentclass[onecolumn,usenatbib]{mn2e}
\usepackage{graphicx}

\input{buote.defs}

\newcommand\apj{{ApJ}}%
\newcommand\apjl{{ApJ}}%
\newcommand\aap{{A\&A}}%
\newcommand\jcap{{J. Cosmology Astropart. Phys.}}%
\newcommand\mnras{{MNRAS}}%
\newcommand\zap{{ZAp}}%
\title[Spherically Averaging Clusters: I. Analytical Relations]{Spherically Averaging Ellipsoidal Galaxy Clusters
in X-Ray and Sunyaev-Zel'dovich Studies: I.\ Analytical Relations}
\author[D.\ A.\ Buote and P.\ J.\ Humphrey]{David A.\
Buote\thanks{E-mail: buote@uci.edu} and
Philip J.\ Humphrey \\ Department of Physics and Astronomy, University
of California, Irvine, 4129 Frederick Reines Hall, Irvine, CA
92697-4575, USA}

\date{Accepted 2011 November 7. Received 2011 November 3; in original form 2011 September 29}

\begin{document} 
\maketitle

\begin{abstract}
This is the first of two papers investigating the deprojection and
spherical averaging of ellipsoidal galaxy clusters. We specifically
consider applications to hydrostatic X-ray and Sunyaev-Zel'dovich (SZ)
studies, though many of the results also apply to isotropic
dispersion-supported stellar dynamical systems. Here we present
analytical formulas for galaxy clusters described by a gravitational
potential that is a triaxial ellipsoid of constant shape and
orientation.  For this model type we show that the mass bias due to
spherically averaging X-ray observations is independent of the
temperature profile, and for the special case of a scale-free
logarithmic potential, there is exactly zero mass bias for any shape,
orientation, and temperature profile. The ratio of spherically
averaged intracluster medium (ICM) pressures obtained from SZ and
X-ray measurements depends only on the ICM intrinsic shape, projection
orientation, and $H_0$, which provides another illustration of how
cluster geometry can be recovered through a combination of X-ray and
SZ measurements.  We also demonstrate that $Y_{\rm SZ}$ and $Y_{\rm
X}$ have different biases owing to spherical averaging, which leads to
an offset in the spherically averaged $Y_{\rm SZ} - Y_{\rm X}$
relation.  A potentially useful application of the analytical formulas
presented is to assess the error range of an observable (e.g., mass,
$Y_{\rm SZ}$) accounting for deviations from assumed spherical
symmetry, without having to perform the ellipsoidal deprojection
explicitly.  Finally, for dedicated ellipsoidal studies, we also
generalize the spherical onion peeling method to the triaxial case for
a given shape and orientation.

\end{abstract}

\begin{keywords}{X-rays: galaxies: clusters ---  X-rays: galaxies --- dark
matter --- cosmological parameters --- cosmology:observations}
\end{keywords}

\newtheorem{defn}{Definition}
\newtheorem{thm}{Theorem}
\newtheorem{coro}{Corollary}
\newtheorem{lem}{Lemma}

\section{Introduction}
\label{intro}

Although it is well-known that galaxy clusters are not spherical,
spherically averaged measurements of cluster properties are standard
practice owing to their comparative simplicity and the expectation
that non-spherical effects do not dominate the error budget. However,
as cosmological measurements with clusters become ever more
precise~\citep[e.g.,][]{vikh09b,alle11a,pier11a}, the need for precise
control of systematic errors also increases.  Deviations from
spherical symmetry, such as intrinsic flattening and substructure,
will necessarily introduce scatter, and possibly biases, into
spherically averaged global scaling relations used in cosmological
studies~\citep[e.g.,][]{evra96a,whit02a,krav06a,shaw08a,krau11a};
e.g., the scaling between total mass and average intracluster medium
(ICM) temperature, quantities that are usually computed interior to a
spherical volume specified by a standard fraction of the virial
radius.

Unfortunately, owing to their increased complexity and computational
expense, non-spherical models are rarely accounted for in detail, if
at all, in the error budgets of cluster measurements.  The modeling of
a cluster possessing substructure requires sophisticated
three-dimensional N-body simulations, while even to compute the
gravitational potential of simple ellipsoidal mass distributions
involves solving complicated integrals~\citep[e.g.,][]{chan87,bt}
unless approximate fitting formulas appropriate for nearly spherical
systems are employed~\citep{lee04a}. The pioneering studies by
\citet{piff03a} and~\citet{gava05a} dedicated to the problem of
assessing errors arising from the assumption of spherical symmetry
employed ellipsoidal models with some simplifying assumptions to
reduce the total computational expense; e.g., \citet{gava05a} studied
the face-on projections of spheroidal NFW~\citep{nfw} mass
distributions with isothermal ICM. These authors found that quantities
obtained from spherical X-ray studies vary typically by $\la 5\%$ as a
result of different intrinsic shapes and viewing orientations for a
cluster.  While not a large effect in absolute terms, errors of a few
percent may be important for precision cosmological studies and
deserve further elucidation.

The non-spherical shapes of clusters need not be considered merely a
nuisance as a source of systematic error, since they are interesting
to study in their own right. With the aid of large, cosmological
N-body simulations significant theoretical progress has been made in
the understanding of the intrinsic shapes of $\Lambda$CDM cluster dark
matter halos
\citep[e.g.,][]{jing95a,mohr95a,spli97a,thom98a,jing02a,bull02a,suwa03a,spri04a,hopk05a,kasu05a,lee05a,allg06a,lee06a,paz06a,ho06a,shaw06a,bett07a,gott07a,macc08a,muno11a,ross11a}.
The axial ratios of dark matter halos in $\Lambda$CDM are found to be
sensitive to the value of $\sigma_8$, the normalization of the power
spectrum of density fluctuations, and are weakly sensitive to the
cosmological matter density parameter, $\Omega_m$. Dark matter halos
in $\Lambda$CDM become flatter with increasing mass, and cluster-mass
halos tend to be prolate/triaxial. Furthermore, halos will be more
spherical if the dark matter particle is self-interacting, which
provides an additional constraint on particle dark matter models,
especially on the galaxy scale~\citep[e.g.,][]{feng09a,feng10a}.

For a few clusters X-ray observations have measured flattened dark
matter halos by adopting hydrostatic models of the
ICM~\citep{frg,buot92a,buot96c}, and there is recent evidence for triaxiality
in the dark matter when X-ray data are combined with constraints from
strong gravitational
lensing~\citep{mora10a,mora11a,mora11b}. Isothermal hydrostatic models
of the ICM applied to the dark matter halos formed in cosmological
simulations predict distributions of average X-ray isophotal axial
ratios consistent with cluster
observations~\citep{wang04a,flor07a,kawa10a}. When baryons are
included in N-body, hydrodynamical simulations of cosmic structure, it
is found that baryon condensation leads to rounder dark matter
halos~\citep[e.g.,][]{dubi94a,kaza04a,deba08a}. However, the treatment
of baryon evolution in cosmological simulations remains a problem for
cluster cores. For example, baryon cooling in some cosmological
simulations can lead to highly flattened ICM cluster cores that
disagree with X-ray observations of relaxed clusters~\citep{fang09a},
whereas the global ICM shapes of the model clusters agree with X-ray
observations and also with the shapes of the total gravitational
potentials of the models as expected for hydrostatic
equilibrium~\citep{buot95a,fang09a,lau11a}. ICM shapes have also been
measured via the thermal Sunyaev-Zel'dovich (SZ) effect~\citep{defi05a,saye11a} which provides
another promising avenue for probing intrinsic cluster shapes.

The non-spherical results cited above represent only a small minority
of X-ray and SZ cluster studies. To a large extent spherical models
dominate because they are easier to implement and are substantially
more computationally efficient than are models which involve a
gravitational potential generated by an ellipsoidal mass
distribution. Here we investigate a different type of model where the
potential, rather than the underlying mass distribution, is an
ellipsoid of constant shape and orientation. These ellipsoidal models
lend themselves to straightforward generalizations of simple, analytic
spherical models, and are therefore just as computationally efficient.
Moreover, we show that many cluster quantities derived assuming
spherical symmetry can be easily interpreted in terms of these
ellipsoidal models for a given shape and orientation.

The paper is organized as follows. In \S \ref{ep} we define the
ellipsoidal models. We show in \S \ref{mass} that the relationship
between the mass profile and potential for these models closely
resembles the spherical case, as does that between the mass profile
and the ICM density and temperature for the case of hydrostatic
equilibrium. In \S \ref{deproj} we provide several analytic
expressions for quantities associated with these ellipsoidal models,
in particular relating their deprojected spherical averages to their
intrinsic ellipsoidal profiles.  For practical use of these models
with binned data, in \S \ref{proj.ep} we generalize the traditional
spherical onion peeling deprojection method appropriate for a series
of concentric, triaxial ellipsoidal shells. Our conclusions are
presented in \S \ref{conc}. In Paper~2~\citep{buot11c} we perform a
detailed investigation of biases and scatter in the measurements of
global quantities resulting from the spherical averaging of
ellipsoidal galaxy clusters. Finally, we mention that the formulas we
present here (and the results in Paper~2) apply not just for massive
clusters but also for groups and massive elliptical galaxies with hot
gaseous halos.

\section{Ellipsoidal Potentials}
\label{ep}

Consider an ellipsoid with principal axes $a,b,c$ and axis ratios, $p_v
\equiv b/a$ and $q_v \equiv c/a$, satisfying $0 < q_v \le p_v \le
1$. When the ellipsoid is aligned so that $a$ lies along the $x$-axis,
$b$ along the $y$-axis, and $c$ along the $z$-axis, then the
ellipsoidal radius $a_v$ is given by,
\begin{equation}
a_v^2 = x^2 + \frac{y^2}{p_v^2} + \frac{z^2}{q_v^2}. \label{eqn.av} 
\end{equation}
We define an ellipsoidal gravitational potential (EP) to be an
ellipsoid of constant shape $(p_v,q_v)$ and orientation so that $\Phi$
depends only on $a_v$; i.e., $\Phi = \Phi(a_v)$. Since $\Phi$ has a
constant shape, EPs best approximate those clusters having nearly
constant ICM shapes. While it is well-known that individual clusters,
both those observed and formed in cosmological simulations, can
sometimes exhibit large radial shape variations, we note that the
average X-ray ellipticity profile for a low-redshift cluster sample
varies only weakly with radius~\citep[see Fig.\ 8 of][]{fang09a}.  In
contrast, the shapes of the isodensity surfaces of the underlying mass
distribution of an EP generally vary with radius.  This can be a
desirable feature since dark matter halos formed in cosmological
simulations typically have radially varying
shapes~\citep[e.g.,][]{jing02a,kaza04a,bail05a,allg06a,vera11a}.

Probably the key advantage of an EP is that a simple, analytic form
for $\Phi$ can be adopted based on a straightforward generalization of
a spherical potential, enabling much faster computational evaluation
than for ellipsoidal mass distributions; e.g., see Paper~2 for EP
generalizations of the NFW and isothermal $\beta$ models. A
disadvantage of an EP is that when the potential is sufficiently
flattened (the amount depending on how steep is the radial potential
profile), the mass density can become negative in some
region~\citep[e.g.,][]{binn81a,blan87a,schn92a,kass93a,evan94a,bt}. Requiring
a non-negative phase-space distribution function (DF) will further
restrict the flattening of the matter distribution that generates the
potential if the DF depends only on the energy and one-component of
the angular momentum~\citep{evan93a,evan94a}, although it is unclear
that the same restrictions would apply to a general three-integral DF.
Moreover, we show below (\S \ref{mass.gen}) that the mass enclosed
within $a_v$ has a form analogous to the mass enclosed within the
spherical radius $r$. It is therefore as well-behaved as the mass
profile of a spherical model, indicating that EPs should be suitable
for many applications. Indeed, in our previous X-ray studies of
elliptical galaxies~\citep{buot96a,buot02b} and a simulated
cluster~\citep{buot95a} we found very good agreement between the
gravitating mass profiles inferred using EPs with those obtained using
the far more computationally expensive ellipsoidal mass distributions.

\section{Mass Distribution}
\label{mass}

\subsection{General Case}
\label{mass.gen}

As noted previously, unlike $\Phi$ itself, the mass density
$\rho(x,y,z)$ of an EP is not an ellipsoid, and its typically
complicated form must be inferred from direct solution of Poisson's
Equation. However, $M(<a_v)$, the mass enclosed within $a_v$, is far
simpler to compute given $\Phi$, as we now show.

Gauss's Law for the mass enclosed within a surface $S$ is,
\begin{equation}
M_{\rm enc} = \frac{1}{4\pi \rm G}\oint_{S}
\nabla\Phi\cdot\hat{n}dS. \label{eqn.gauss}
\end{equation}
Now assuming $\Phi = \Phi(a_v)$ for an EP and taking $S$ to be the
ellipsoidal surface defined by $a_v$, the gradient and vector surface
element take the form,
\begin{equation}
\nabla\Phi(a_v) =\left(x\hat{i} +
\frac{y}{p_v^2}\hat{j} + \frac{z}{q_v^2}\hat{k}\right)
\frac{1}{a_v}\frac{d\Phi}{da_v} , \:\:\:\: {\rm and} \:\:\:\:
\hat{n}{dS} = \left(\frac{xq_v\hat{i}}{\sqrt{a_v^2 - x^2 - y^2/p_v^2}} +
\frac{yq_v\hat{j}}{p_v^2\sqrt{a_v^2 - x^2 - y^2/p_v^2}} +
\hat{k}\right)dxdy.  \label{eqn.gradient} 
\end{equation}
Substituting these expressions into eqn.\ (\ref{eqn.gauss}) and
rearranging into two separate surface integrals yields, 
\begin{equation}
M(<a_v) = \frac{1}{4\pi \rm G}\frac{q_v}{a_v}\frac{d\Phi}{da_v}\int\int
\frac{x^2\left(1- 1/q_v^2\right)
+ a_v^2/q_v^2}{\sqrt{a_v^2 - x^2 - y^2/p_v^2}}dxdy +
\frac{1}{4\pi \rm G}\frac{q_v}{a_v}\frac{d\Phi}{da_v}\int\int
\frac{y^2}{p_v^2}\left(\frac{1}{p_v^2} - \frac{1}{q_v^2}\right) 
\frac{1}{\sqrt{a_v^2 - x^2 - y^2/p_v^2}}dxdy. \label{eqn.gauss.long}
\end{equation}
Because of the ellipsoidal symmetry we may evaluate the integrals in
any quadrant of the $x-y$ plane. Multiplying the result obtained from
one quadrant by four gives the total surface integral in the positive
$z$-direction, and multiplying by eight gives the result for the total
ellipsoidal surface. By effecting a change of variable, $y =
up_v\sqrt{a_v^2 - x^2}$, the first term on the R.H.S.\ in eqn.\
(\ref{eqn.gauss.long}) simplifies to,
\begin{equation}
\frac{2}{\pi \rm G}\frac{p_vq_v}{a_v}\frac{d\Phi}{da_v} 
\int_0^{a_v}dx \left[x^2\left(1 - \frac{1}{q_v^2}\right) + \frac{a_v^2}{q_v^2}\right] 
\int_0^1 \frac{du}{\sqrt{1 - u^2}} =  \frac{a_v^2}{\rm G}\frac{d\Phi}{da_v} 
\frac{p_vq_v}{3}\left(1 + \frac{2}{q_v^2}\right). \label{eqn.gauss.first}
\end{equation}
Here the $u$-integral evaluates to $\pi/2$ so that the $x$-integral is
elementary. Similarly, making the same change of variable in the
second term on the R.H.S.\ in eqn.\ (\ref{eqn.gauss.long}) yields,
\begin{equation}
\frac{2}{\pi \rm G}\frac{p_vq_v}{a_v}\frac{d\Phi}{da_v}
\left(\frac{1}{p_v^2} - \frac{1}{q_v^2}\right)
\int_0^{a_v}\left(a_v^2 - x^2\right) dx \int_0^1 \frac{u^2du}{\sqrt{1 - u^2}} =   
\frac{a_v^2}{\rm G}\frac{d\Phi}{da_v} 
\frac{p_vq_v}{3}\left(\frac{1}{p_v^2} - \frac{1}{q_v^2}\right), \label{eqn.gauss.second}
\end{equation}
where this time the $u$-integral evaluates to $\pi/4$, again leading
to an elementary $x$-integral. Substituting the results from
equations (\ref{eqn.gauss.first}) and (\ref{eqn.gauss.second}) back
into eqn.\ (\ref{eqn.gauss.long}) gives the result for the enclosed
mass, which we state as a theorem:
\begin{thm} \label{thm.epmass}
For an EP the mass enclosed within the ellipsoid of radius $a_v$ is,
\begin{equation}
M(<a_v) = \eta (p_v,q_v) \frac{a_v^2}{\rm G}\frac{d\Phi}{da_v}, \label{eqn.mep}
\end{equation}
where we define the EP shape factor,
\begin{equation}
\eta (p_v,q_v) \equiv \frac{p_v q_v}{3}\left[1 + \frac{1}{p_v^2} +
\frac{1}{q_v^2} \right], \label{eqn.eta}
\end{equation}
which is of order unity. 
\end{thm}
Hence, the relation between $M(<a_v)$ and $\Phi(a_v)$ is nearly
identical in form to the spherical case, $M(<r) = (r^2/{\rm
G})d\Phi/dr$, to which it reduces for $p_v=q_v=1$. Here it is worth
emphasizing that the axis ratios $p_v$ and $q_v$ define the shape of
$\Phi$ and the bounding surface for the calculation of $M(<a_v)$, but
generally the axis ratios of the corresponding mass density
distribution are smaller (e.g., see Fig.\ 1 of Paper~2).

The spherically averaged mass distribution of an EP does not possess
such a simple relationship to the potential. Nevertheless, it can be
computed without resorting to the evaluation of the volume integral of
the generally complicated density distribution.
\begin{thm} \label{thm.epspmass}
For an EP the mass enclosed within the sphere of radius $r$ is,
\begin{equation}
M(<r) = \frac{r}{4\pi \rm G}\int_{4\pi}a_v\frac{d\Phi (a_v)}{da_v}d{\rm
\Omega} = \frac{r}{4\pi \eta (p_v,q_v)}\int_{4\pi}\frac{M(<a_v)}{a_v}d{\rm
\Omega},
\end{equation}
where $d{\Omega}=\sin\theta d\theta d\phi$ is the solid angle and the
integration proceeds over the entire spherical surface.
\end{thm}
\begin{proof} 
Again using Gauss's Law (eqn.\ \ref{eqn.gauss}), but this time for a
spherical surface, gives,
\begin{equation}
\hat{n}{dS} = \hat{\rm r} {r^2d\Omega} = \left( x\hat{i} +
y\hat{j} + z\hat{k}\right)r d\Omega,
\end{equation}
for the vector surface element. Taking the dot product with $\nabla\Phi(a_v)$
(eqn.\ \ref{eqn.gradient}) and integrating over the sphere gives the
stated result, where the R.H.S.\ made use of Theorem
\ref{thm.epmass}.  
\end{proof}

\subsection{Hydrostatic Equilibrium}
\label{mass.he}

Since a primary goal of our study is to investigate the effect of
spherical averaging on the inferred mass distribution from X-ray
observations, we need to consider the case where the hot intracluster
medium (ICM), or ``hot gas'', is in hydrostatic equilibrium. We assume
the self-gravity of the gas can be neglected $(\Phi_{\rm gas} \ll
\Phi))$ which is generally a very good approximation interior to
$r_{500}$ where the cluster gas fraction is $\sim
10\%$~\citep[e.g.,][]{prat10a}. In this case, the equation of
hydrostatic equilibrium,
\begin{equation}
\nabla P_{\rm gas} = -\rho_{\rm gas}\nabla\Phi, \label{eqn.he}
\end{equation}
where $P_{\rm gas}$ is the thermal pressure and $\rho_{\rm gas}$ is
the density of the ICM, requires that surfaces of constant potential
are the same as surfaces of constant ICM pressure, density,
temperature, and, so long as the metal abundances do not vary over
these surfaces, it follows that the surfaces of constant X-ray
emissivity also follow the potential~\citep[``X-ray Shape
Theorem,''][]{buot94,buot96a,buot11a}. For the special case of an EP
model it follows that all ICM quantities depend only on $a_v$; e.g.,
$P_{\rm gas} = P_{\rm gas}(a_v)$, $\rho_{\rm gas} = \rho_{\rm
gas}(a_v)$, and $T = T(a_v)$, so that the hydrostatic equation
becomes,
\begin{equation} 
\frac{dP_{\rm gas}(a_v)}{da_v} = -\rho_{\rm gas}(a_v)\frac{d\Phi
(a_v)}{da_v}, \label{eqn.he.av}
\end{equation}
where we have used the definition of the gradient in eqn.\
(\ref{eqn.gradient}). Using Theorem \ref{thm.epmass} and substituting
the ideal gas equation of state for the pressure, $P_{\rm gas} =
\rho_{\rm gas} k_B T/(\mu m_{\rm a})$, where $k_B$ is Boltzmann's
constant, $m_{\rm a}$ is the atomic mass unit, and $\mu$ is the mean
atomic weight of the gas, we obtain the following
result.
\begin{thm} \label{thm.trad}
For an EP in hydrostatic equilibrium the mass enclosed within an
ellipsoid of radius $a_v$ is,
\begin{equation}
M(<a_v) = -\eta (p_v,q_v)\left[ \frac{a_v{\rm k_B}T}{\mu\rm m_aG}\right] \left[
\frac{d\ln\rho_{\rm gas}}{d\ln a_v} + \frac{d\ln T}{d\ln a_v} \right],
\label{eqn.trad.ep}
\end{equation}
where $\eta$ is given by eqn.\ (\ref{eqn.eta}). 
\end{thm}
This result is a simple generalization of the spherical case where $r$
is replaced by $a_v$ and the mass is multiplied by the shape factor
$\eta (p_v,q_v)$. Similarly, solutions of eqn.\ (\ref{eqn.he.av}) for
$\rho_{\rm gas}$, $T$, and the entropy are also easily constructed
from their spherical counterparts by making the transformation
$r\rightarrow a_v$ and $M(<r)\rightarrow M(<a_v)/\eta$~\citep[e.g.,
see equations 10-12 of][]{buot11a}.

\section{Deprojection Relations for Spherically Averaged Quantities}
\label{deproj}

Here we derive analytical expressions for the spherically averaged
deprojection of intrinsically ellipsoidal quantities, with a
particular application to the EPs.  We begin by summarizing the
general results for the projection of a triaxial ellipsoid. 

\subsection{Preliminaries}

We define the orientation of the ellipsoidal system
following~\citet{binn85}. In the reference $(x,y,z)$ coordinate system
the principal axes of the ellipsoid are aligned with the coordinate
directions as described in the definition of $a_v$ (eqn.\
\ref{eqn.av}). Define a new rotated coordinate system
$(x^{\prime},y^{\prime} ,z^{\prime} )$ where the $z^{\prime}$ axis
lies along the line-of-sight to the observer from the center of the
ellipsoid and the $x^{\prime}$ axis is located in the $(x,y)$ plane.
The two systems are related by starting with their axes aligned,
rotating the reference system first by an angle $\phi$ about the
$z^{\prime}$ axis and then by an angle $\theta$ about the $x^{\prime}$
axis. (For the rotation matrix, see \citealt{binn85}.)

For any quantity that depends only on the ellipsoidal radius $a_v$,
\citet{cont56} and \citet{star77} showed that the projection of this
ellipsoidal quantity yields a (two-dimensional) elliptical
distribution of constant shape and orientation that depends only on
the elliptical coordinate $a_s$ on the sky; i.e., the projection of
the ellipsoidal volume emissivity $\epsilon(a_v)$ (i.e., luminosity
per volume) yields an elliptical surface brightness $\Sigma(a_s)$,
\begin{equation}
\Sigma(a_s) = \frac{2}{\sqrt{f}}\int_{a_s}^{\infty}
\frac{\epsilon(a_v)a_vda_v}{\sqrt{a_v^2 - a_s^2}}, \label{eq.ep.proj}
\end{equation}
where,
\begin{equation}
f = \sin^2\theta\left(\cos^2\phi + \frac{\sin^2\phi}{p_v^2}\right) +
\frac{\cos^2\theta}{q_v^2}, \label{eqn.f}
\end{equation}
using the angle definitions described above, and where $p_v$ and $q_v$
define $a_v$. The elliptical coordinate variable $a_s=\gamma_s
q_s\alpha$ is proportional to the elliptical radius $\alpha$ (i.e.,
semi-major axis on the sky), which is of more immediate interest to the observer,
\begin{equation}
\alpha^2 = X^2 + \frac{Y^2}{q_s^2}, \label{eqn.alpha}
\end{equation}
where $q_s$ is the axial ratio of $\Sigma$, and $X$ and $Y$ are sky
coordinates aligned with the isophotal major and minor axes
respectively\footnote{To convert our notation to that used by
\citet{star77} for the projected quantities, let $\gamma_s\rightarrow
\alpha$, $q_s\rightarrow \beta$, $\alpha\rightarrow q$, and
$X\Leftrightarrow Y$.}. The proportionality factors are given by,
\begin{eqnarray}
\gamma_s^2 & = &\frac{1}{2f}\left[\left(\rm A+C\right) + \sqrt{\left(\rm
A-C\right)^2 + \rm B^2}\right],\\
q_s^2 & = & \frac{\left(\rm A+C\right) - \sqrt{\left(\rm
A-C\right)^2 + \rm B^2}}{\left(\rm A+C\right) + \sqrt{\left(\rm
A-C\right)^2 + \rm B^2}} = \left(\sqrt{f}\gamma_s^2 p_vq_v\right)^{-2}, \label{eqn.qs}
\end{eqnarray}
which were derived by \citet{star77}, though also see \citet{bing80a} and \citet{binn85}, where
\begin{eqnarray}
A & = & \frac{\cos^2 \theta}{q_v^2}\left(\sin^2\phi +
\frac{\cos^2\phi}{p_v^2}\right) + \frac{\sin^2\theta}{p_v^2},\\
B & = & \cos\theta\sin 2\phi\left(1
-\frac{1}{p_v^2}\right)\frac{1}{q_v^2},\\
C & = & \left(\frac{\sin^2\phi}{p_v^2} + \cos^2\phi\right)\frac{1}{q_v^2}.
\end{eqnarray}

\subsection{General Case}

\begin{defn} \label{circavg}
For an elliptical distribution $\Sigma(\alpha)$, where $\alpha$ is the
elliptical radius with axial ratio $q_s$, we define the ``effective
circular average'' $\langle \Sigma(R)\rangle$ by associating the
circular radius with the geometric mean radius: $\langle
\Sigma(R)\rangle = \Sigma(\alpha)$, where $R=\alpha\sqrt{q_s}$.
\end{defn}

The use of such an effective circular average is required for
obtaining the analytical relations we describe below. We note that for
typical models (i.e., those considered in Paper~2) the effective
circular average usually very closely approximates a formal azimuthal
integration of $\Sigma(\alpha)$ at fixed $R$. Only for very flattened
models that fall steeply with radius do we find notable differences.
We emphasize, however, that for our purposes it is only necessary that
the observer employs the effective circular average in their analysis
(as is a very common practice).

\begin{defn} \label{sphdep}
For any quantity intrinsically distributed as an
ellipsoid stratified on surfaces of constant ellipsoidal radius $a_v$,
let $\langle\cdots\rangle^{\rm d}$ represent the deprojected,
spherically averaged version of that quantity defined in the following
manner. 
For illustration consider the specific case of $\epsilon(a_v)$ and its
projection $\Sigma(a_s)$ as defined previously.  The spherical
averaging considered here begins by first taking the effective
circular average of $\Sigma(a_s)$ (i.e., Definition~\ref{circavg} and
using the fact that $a_s\propto\alpha$) to give the radial profile,
$\langle\Sigma(R)\rangle$, where $R$ is the radius on the sky. Then
$\langle\Sigma(R)\rangle$ is deprojected by assuming spherical
symmetry, which yields the spherically averaged volume emissivity
denoted by, $\langle\epsilon(r)\rangle^{\rm d}$.
\end{defn}

\begin{defn} \label{sphdep.comp}
For any quantity that is composed of one or more deprojected
spherically averaged quantities as defined in Definition~\ref{sphdep},
but itself does not satisfy Definition~\ref{sphdep}, we shall also
refer to this composite quantity as a ``deprojected spherical
average'' but add a plus sign in the notation to distinguish it from
Definition~\ref{sphdep}; i.e., $\langle\cdots\rangle^{\rm d+}$.
\end{defn}

\begin{thm} \label{thm.general}
For any ellipsoidal distribution that depends only on
the ellipsoidal radius $a_v$, the deprojected spherical average
(Definition~\ref{sphdep}) of this distribution is,
\begin{equation}
\langle\epsilon(r)\rangle^{\rm d}  = 
\left(\gamma_s q_s^{\frac{1}{2}} f^{-\frac{1}{2}}\right)
\epsilon(a_v), \;{\rm where} \; a_v = (\gamma_s\sqrt{q_s})r.
\end{equation} 
\end{thm}

\begin{proof}
We start with $\epsilon(a_v)$ and use eqn.\ (\ref{eq.ep.proj}) to
compute $\Sigma(a_s)$. From Definition~\ref{circavg} the effective
circular average of $\Sigma(a_s)$ is obtained by associating the
circular radius $R$ with the geometric mean radius of the ellipse of
semi-major axis $\alpha$, $R = \alpha\sqrt{q_s}$.  Since the
ellipsoidal coordinate $a_s$ is a function of the elliptical radius,
$a_s(\alpha)=\gamma_s q_s\alpha$, it follows that,
\begin{equation}
\langle\Sigma (R)\rangle = \Sigma (a_s(\alpha)) = \Sigma
(a_s(\frac{R}{\sqrt{q_s}})) = \Sigma (\gamma_s \sqrt{q_s} R),
\end{equation}
which relates the circular distribution on the L.H.S.\ to the
elliptical distribution on the R.H.S. Since $\langle\Sigma (R)\rangle$
is circularly symmetric, it may be spherically deprojected using the inverse Abel
integral relation,
\begin{equation}
\langle\epsilon(r)\rangle^{\rm d}  =  -\frac{1}{\pi}
\int_r^{\infty} \frac{d \langle\Sigma (R)\rangle}{dR} \frac{dR}{\sqrt{
R^2 - r^2}}.
\end{equation}
Changing the integration variable from $R$ to $a_s$ yields,
\begin{equation}
\langle\epsilon(r)\rangle^{\rm d}  =  -\frac{\gamma_s \sqrt{q_s}}{\pi}
\int_{a_v}^{\infty} \frac{d \Sigma (a_s)}{da_s} \frac{da_s}{\sqrt{
a_s^2 - a_v^2}}, \label{eqn.abel1}
\end{equation}
where $a_v = (\gamma_s \sqrt{q_s}) r$. Since eqn.\ (\ref{eq.ep.proj})
is an Abel integral, its inverse is readily obtained~\citep[e.g.,
eqn.\ 13 of][]{star77},
\begin{equation}
\epsilon(a_v)  =  -\frac{\sqrt{f}}{\pi}
\int_{a_v}^{\infty} \frac{d \Sigma (a_s)}{da_s} \frac{da_s}{\sqrt{
a_s^2 - a_v^2}}. \label{eqn.abel2}
\end{equation}
Comparing equations (\ref{eqn.abel1}) and (\ref{eqn.abel2}) gives the
desired result.
\end{proof}

\subsection{X-Ray Emission and Hydrostatic Equilibrium}

Here we consider cluster properties associated with the ICM X-ray
emission. We assume that all volume ICM properties depend only on
ellipsoidal radius $a_v$, which applies exactly for hydrostatic
equilibrium in an EP (see \S \ref{mass.he}). However, hydrostatic
equilibrium is strictly required below only for Theorems~\ref{thm.m}
and~\ref{thm.scalefree} and Corollary~\ref{coro.fgas}.

\begin{thm} \label{thm.temp}
The deprojected spherical average of the emission-weighted temperature is,
\begin{equation}
\langle T(r)\rangle^{\rm d+} =  T(a_v), \;{\rm where} \; a_v = (\gamma_s\sqrt{q_s})r.
\end{equation}
\end{thm}
\begin{proof}
Let $\epsilon$ be the X-ray emissivity and $T$ the gas
temperature. The emission-weighted temperature is defined as the
volume integral of $\epsilon T$ divided by the volume integral of
$\epsilon$. Hence, at any radius, 
\begin{equation}
\langle T(r)\rangle^{\rm d+} =  { \langle (\epsilon
T)(r)\rangle^{\rm d} \over \langle \epsilon (r)\rangle^{\rm d}}
=  \frac{\left(\gamma_s q_s^{\frac{1}{2}} f^{-\frac{1}{2}}\right)
\epsilon (a_v) T (a_v)}{ \left(\gamma_s q_s^{\frac{1}{2}}
f^{-\frac{1}{2}}\right)\epsilon (a_v)}
=  T (a_v), 
\end{equation}
where $a_v = (\gamma_s\sqrt{q_s})r,$ and we made use of Theorem
\ref{thm.general} in both the numerator and denominator.
Note that the deprojected spherical averages in the numerator and
denominator each correspond to Definition \ref{sphdep}, while the
composite result $\langle T(r)\rangle^{\rm d+}$ corresponds to
Definition \ref{sphdep.comp}. Similar usages apply below.
\end{proof}

\begin{thm} \label{thm.rhog}
The deprojected spherical average of the gas density is,
\begin{equation}
\langle\rho_{\rm gas}(r)\rangle^{\rm d+} =
\left(\gamma_s^{\frac{1}{2}} q_s^{\frac{1}{4}} f^{-\frac{1}{4}}\right) \rho_{\rm gas}(a_v), \;{\rm where} \; a_v = (\gamma_s\sqrt{q_s})r.
\end{equation}
\end{thm}
\begin{proof}
From the definition of the X-ray emissivity,
\begin{equation}
\langle\rho_{\rm gas}(r)\rangle^{\rm d+} = \left[
\frac{\langle\epsilon(r)\rangle^{\rm d}}{\Lambda (\langle
T(r)\rangle^{\rm d+}) } \right]^{\frac{1}{2}},
\end{equation}
where we have suppressed the metallicity dependence of the plasma
emissivity $\Lambda$, which does not affect our arguments provided
that the metallicity depends only on $a_v$ as we are assuming for all
ICM properties. Applying Theorem \ref{thm.general} in the numerator
and Theorem \ref{thm.temp} in the denominator of the above equation
yields,
\begin{equation}
\langle\rho_{\rm gas}(r)\rangle^{\rm d+} = \left[
\frac{ \left(\gamma_s q_s^{\frac{1}{2}}
f^{-\frac{1}{2}}\right)\epsilon(a_v)}{\Lambda ( T(a_v) ) }
\right]^{\frac{1}{2}} = \left[
\frac{ \left(\gamma_s q_s^{\frac{1}{2}}
f^{-\frac{1}{2}}\right) \rho_{\rm gas}^2(a_v)\Lambda(T(a_v))}{\Lambda ( T(a_v) ) }
\right]^{\frac{1}{2}},
\end{equation}
where $a_v = (\gamma_s\sqrt{q_s})r$, which reduces to the desired result.
\end{proof}

\begin{coro} \label{coro.pgasentropy}
The deprojected spherical averages of the gas pressure and entropy are,
\begin{equation}
\langle P_{\rm gas}(r)\rangle^{\rm d+} =
\left(\gamma_s^{\frac{1}{2}} q_s^{\frac{1}{4}} f^{-\frac{1}{4}}\right)
P_{\rm gas}(a_v), \:\:\:\:\:\: 
\langle S(r)\rangle^{\rm d+} =
\left(\gamma_s^{-\frac{1}{3}} q_s^{-\frac{1}{6}} f^{\frac{1}{6}}\right)
S(a_v), \;{\rm where} \; a_v = (\gamma_s\sqrt{q_s})r.
\end{equation}
\end{coro}
\begin{proof}
These results are immediate consequences of the definitions of each
quantity, $P_{\rm gas}= \rho_{\rm gas} k_BT/(\mu m_a)$ and $S=
(k_B/\mu m_a)T \rho_{\rm gas}^{-2/3}$, Theorem \ref{thm.temp},
and Theorem \ref{thm.rhog}.
\end{proof}

\begin{thm} \label{thm.m}
The deprojected spherical average of the total mass enclosed within
radius $r$ is, 
\begin{equation}
\langle M(<r)\rangle^{\rm d+} =  \left(\gamma_s q_s^{\frac{1}{2}}
\eta(p_v,q_v)\right)^{-1} M(<a_v), \;{\rm where} \; a_v =
(\gamma_s\sqrt{q_s})r, 
\end{equation}
for any temperature profile.
\end{thm}
\begin{proof}
Applying hydrostatic equilibrium for a spherically symmetric cluster
gives, 
\begin{equation}
\langle M(<r)\rangle^{\rm d+} = \frac{-1}{\langle\rho_{\rm
gas}(r)\rangle^{\rm d+}} \frac{r^2}{{\rm G}}  \frac{d}{dr}\langle
P_{\rm gas}(r)\rangle^{\rm d+} = \frac{-1}{\rho_{\rm gas}(a_v)}
\frac{r^2}{{\rm G}} \frac{d}{dr}P_{\rm gas}(a_v),
\end{equation}
where $a_v = (\gamma_s\sqrt{q_s})r$, and the R.H.S.\ made use of
Corollary \ref{coro.pgasentropy}, Theorem \ref{thm.rhog}, and the fact
that $\gamma_s$, $q_s$, and $f$ depend only on $p_v$, $q_v$, and the
fixed line-of-sight projection orientation.  Changing the variable
from $r$ to $a_v = (\gamma_s\sqrt{q_s})r$ so that $d/dr =
(\gamma_s\sqrt{q_s})d/da_v$, after simplifying, gives,
\begin{equation}
\langle M(<r)\rangle^{\rm d+} =  \frac{1}{\gamma_s\sqrt{q_s}} \left[
\frac{-1}{\rho_{\rm gas}(a_v)} \frac{a_v^2}{{\rm G}}
\frac{d}{da_v}P_{\rm gas}(a_v) \right].
\end{equation}
By making use of eqn.\ (\ref{eqn.he.av}) and Theorem \ref{thm.epmass} the quantity in brackets equals
$M(<a_v)/\eta(p_v,q_v)$, which proves the theorem, for any temperature
profile $T(a_v)$.
\end{proof}

We have chosen to state explicitly that this result for the total mass
holds for any temperature profile since we desire to emphasize this
point below in Theorem~\ref{thm.scalefree}. Next, however, we consider
the gas mass.

\begin{thm} \label{thm.mgas}
The deprojected spherical average of the gas mass enclosed within radius $r$ is,
\begin{equation}
\langle M_{\rm gas}(<r)\rangle^{\rm d+} =
\left(\gamma_s^{\frac{5}{2}} q_s^{\frac{5}{4}} f^{\frac{1}{4}} p_vq_v\right)^{-1} M_{\rm gas}(<a_v), \;{\rm where} \; a_v = (\gamma_s\sqrt{q_s})r.
\end{equation}
\begin{proof}
The gas mass within a spherical volume of radius $r$ is,
\begin{equation}
\langle M_{\rm gas}(<r)\rangle^{\rm d+} = \int_0^r \langle
\rho_{\rm gas}(r)\rangle^{\rm d+} 4\pi r^2dr = 
\left(\gamma_s^{\frac{1}{2}} q_s^{\frac{1}{4}} f^{-\frac{1}{4}}\right)
\int_0^r \rho_{\rm gas}(a_v) 4\pi r^2dr,
\end{equation}
where $a_v = (\gamma_s\sqrt{q_s})r$, and in the R.H.S.\ we applied
Theorem \ref{thm.rhog} and again (as in Theorem~\ref{thm.m}) made use
of the fact that $\gamma_s$, $q_s$, and $f$ do not depend on
$r$. Changing the integration variable from $r$ to $a_v =
(\gamma_s\sqrt{q_s})r$, and simplifying, gives,
\begin{equation}
\langle M_{\rm gas}(<r)\rangle^{\rm d+} = 
\left(\gamma_s^{\frac{5}{2}} q_s^{\frac{5}{4}}
f^{\frac{1}{4}}\right)^{-1} \int_0^{a_v} \rho_{\rm gas}(a_v) 4\pi a_v^2da_v.
\end{equation}
Since an ellipsoidal volume element is $dV = 4\pi p_vq_va_v^2da_v$,
the integral equals $M_{\rm gas}(<a_v)/(p_vq_v)$, which proves the
theorem. 
\end{proof}
\end{thm}

\begin{coro} \label{coro.fgas}
The deprojected spherical average of the gas mass fraction enclosed
within radius $r$ is, 
\begin{equation}
\langle f_{\rm gas}(<r)\rangle^{\rm d+} =
\left(\gamma_s^{\frac{3}{2}} q_s^{\frac{3}{4}}
f^{\frac{1}{4}}p_vq_v/\eta(p_v,q_v)\right)^{-1} f_{\rm gas}(<a_v), \;{\rm
where} \; a_v = (\gamma_s\sqrt{q_s})r. 
\end{equation}
\end{coro}
\begin{proof}
This result follows immediately from the definition of the gas
fraction, $f_{\rm gas}(<r) = M_{\rm gas}(<r)/M(<r)$, Theorem
\ref{thm.m}, and Theorem \ref{thm.mgas}.
\end{proof}

\begin{thm} \label{thm.scalefree} 
For the scale-free logarithmic EP, $\Phi(a_v) = ({\rm
G}M_{\Delta}/a_{\Delta})\ln (a_v)$, spherical averaging does not bias
the mass profile in the sense that,
\begin{equation}
\langle M(<r)\rangle^{\rm d+} = \langle M(<r) \rangle^{\rm true},
\end{equation}
where $\langle M(<r) \rangle^{\rm true}$ is the spherical average of
true mass distribution corresponding to $\Phi(a_v)$. This result is
independent of the gas temperature profile.
\end{thm}
\begin{proof}
This is a special case of the general problem that is the focus of
Paper~2; i.e., we wish to compare spherically averaged quantities
obtained by an observer to those obtained by theoretical
studies. While an observer will measure deprojected spherically
averaged quantities (Definitions \ref{sphdep} and
\ref{sphdep.comp}), the theorist typically spherically averages the
true three-dimensional distribution directly. From Theorem
\ref{thm.epspmass}, we obtain the spherical average of the mass distribution
generated by the scale-free logarithmic EP, $\Phi(a_v) = ({\rm
G}M_{\Delta}/a_{\Delta})\ln (a_v)$, where $M_{\Delta}$ and
$a_{\Delta}$ are constants,
\begin{equation}
\langle M(<r) \rangle^{\rm true} = \frac{r}{4\pi \rm G}\int_{4\pi}a_v\frac{d\Phi (a_v)}{da_v}d{\rm
\Omega} = M_{\Delta}\frac{r}{a_{\Delta}}.
\end{equation}
For comparison, the deprojected spherical average of the mass profile
(Theorem \ref{thm.m}) depends on the mass enclosed within ellipsoidal
radius $a_v$, which is obtained by inserting the definition of the
scale-free $\Phi(a_v)$ into eqn.\ (\ref{eqn.mep}) of
Theorem~\ref{thm.epmass},
\begin{equation}
M(<a_v) = \eta(p_v,q_v)M_{\Delta}\frac{a_v}{a_{\Delta}} =
\left(\gamma_sq_s^{\frac{1}{2}}\eta(p_v,q_v)\right)
M_{\Delta}\frac{r}{a_{\Delta}}, 
\end{equation}
where the R.H.S.\ made use of the substitution $a_v =
(\gamma_s\sqrt{q_s})r$ as appropriate for Theorem \ref{thm.m}. Now
substituting $M(<a_v)$ into Theorem~\ref{thm.m} gives, $\langle
M(<r)\rangle^{\rm d+} = M_{\Delta}r/a_{\Delta} = \langle M(<r)
\rangle^{\rm true}$, independent of the gas temperature profile, which
proves the theorem.
\end{proof}

This theorem complements and extends the result presented in Appendix
B of \citet{chur08a}. These authors consider the bias due to spherical
averaging of an isothermal ICM with a scale-free gas density,
$\rho_{\rm gas}= h(\theta,\phi)r^{-\alpha}$, where $h(\theta,\phi)$ is
some positive function. They argue that this gas density distribution
leads to an inferred total mass that also displays no bias due to
spherical averaging. Because their model implies a potential,
$\Phi\propto \ln [h(\theta,\phi)r^{-\alpha}]$, equivalent to the
scale-free potential we employed above in Theorem \ref{thm.scalefree}
for $\alpha=1$ and $h=1/\sqrt{f}$, it is reassuring that the two
results each predict no bias for the special case of an isothermal
ICM. The assumption of the scale-free EP with no {\it a priori}
restriction on the form of the gas density has allowed us to
generalize rigorously the zero-bias result for any ICM temperature
profile.


\subsection{Sunyaev-Zel'dovich Effect and Related X-ray Quantities}
\label{sz}

Now we consider a galaxy cluster also to be observed via the thermal
SZ effect, and we continue to assume that all
three-dimensional ICM properties depend only on the ellipsoidal radius
$a_v$.

\begin{coro} \label{pressure.sz}
The deprojected spherical average of the ICM electron pressure,
$P_e=n_ek_BT$, obtained from a measurement of the thermal SZ effect
is,
\begin{equation}
\langle P_e (r)\rangle^{\rm d}_{\rm SZ}  = 
\left(\gamma_s q_s^{\frac{1}{2}} f^{-\frac{1}{2}}\right)
P_e (a_v), \;{\rm where} \; a_v = (\gamma_s\sqrt{q_s})r. \label{eqn.psz}
\end{equation} 
\end{coro}

\begin{proof}
The thermal SZ effect is the temperature decrement, $\delta T/T=-2y_c$
in the Rayleigh-Jeans tail of the Cosmic Background Radiation (CBR)
spectrum due to inverse Compton scattering of CBR photons by energetic
ICM electrons. The Compton-y parameter is,
\begin{equation}
y_c = \frac{\sigma_T}{m_e c^2}\int_{\rm los}P_edz^{\prime} =
\frac{\sigma_T}{m_e c^2} \frac{2}{\sqrt{f}}\int_{a_s}^{\infty}
\frac{P_e (a_v)a_vda_v}{\sqrt{a_v^2 - a_s^2}}, 
\end{equation}
where $\sigma_T$ is the Thomson cross section, $m_e$ is the electron
mass, and the R.H.S. follows from the condition that $P_e = P_e(a_v)$
(i.e., eqn.\ \ref{eq.ep.proj}). By associating $y_c(a_s)$ with
$\Sigma(a_s)$ and $\sigma_T P_e(a_v)/(m_c c^2)$ with $\epsilon(a_v)$
the result follows immediately from Theorem \ref{thm.general}. 
\end{proof}

The geometrical factor $\gamma_s q_s^{\frac{1}{2}} f^{-\frac{1}{2}}$
in eqn.\ (\ref{eqn.psz}) is the square of the corresponding factor for
the ICM pressure obtained from X-ray studies (Corollary
\ref{coro.pgasentropy}), indicating that spherical averaging has a
stronger impact on the ICM pressure inferred from SZ studies. Since
this difference in geometrical factors can be exploited to measure the
intrinsic shape and orientation of a cluster, we state it formally.

\begin{coro} \label{pressure.szx}
The ratio of the deprojected spherical averages of the ICM electron
pressures inferred from SZ and X-ray studies is,
\begin{equation}
\langle  P_{\rm SZ,X}(r)\rangle^{\rm d+}  \equiv \frac{\langle P_e (r)\rangle^{\rm d}_{\rm SZ}}{\langle P_e (r)\rangle^{\rm d+}_{\rm X}} = \gamma_s^{\frac{1}{2}} q_s^{\frac{1}{4}}
f^{-\frac{1}{4}}.
\end{equation} 
\end{coro}

\begin{proof}
This result follows immediately from Corollaries
\ref{coro.pgasentropy} and \ref{pressure.sz}, where we have used the
ICM electron pressure from X-rays, $\langle P_e (r)\rangle^{\rm
d+}_{\rm X} = ((2+\mu)/5)\langle P_{\rm gas}(r)\rangle^{\rm
d+}$. \label{eqn.pszx}
\end{proof}

We remark that $\langle P_{\rm SZ,X}(r)\rangle^{\rm d+}$ is, in fact,
constant with radius and does not overtly display any dependence on
the distance to the cluster, and hence the Hubble Constant,
$H_0$. However, $P_{\rm X} \propto \rho_{\rm gas}\propto
\sqrt{\epsilon}$, where $\epsilon$ is the volume emissivity.
To obtain physical units for the emissivity requires converting the
observed X-ray flux to a luminosity density, the net result of which is that
$P_{\rm X}$ is inversely proportional to the square root of
the cluster distance, leading to $\langle P_{\rm SZ,X}(r)\rangle^{\rm
d+}\propto 1/\sqrt{H_0}$.

The possibility of uncovering the intrinsic shape of the cluster ICM
by combining X-ray and SZ measurements has been recognized for over
ten
years~\citep[e.g.,][]{zaro98a,fox02a,rebl00a,lee04a,defi05a,puch06a,sere07a,mahd11a}. This
promising technique has already provided interesting constraints on
cluster shapes for many clusters~\citep{defi05a,sere06a} using
isothermal triaxial $\beta$ models for the ICM, and more recently
models with a radially varying temperature profile have been applied
to the cluster A~1689~\citep{sere11a}. Corollary
\ref{pressure.szx} defines a particular approach to this problem
that has some attractive characteristics. First, the ratio of
spherically averaged pressures is valid for any temperature profile
$T(a_v)$. Second, the relationship does not assume a particular ICM
radial density profile (e.g., $\beta$ model) and, in principle, can be
deprojected using the traditional spherical onion peeling
procedure. Hence, studies can be conducted entirely in the context of
spherical symmetry to obtain $\langle P_{\rm SZ,X}(r)\rangle^{\rm d+}$
and then, supplemented with a measurement of the average ICM axial
ratio on the sky ($q_s$), can be used to constrain the geometrical
factor $ \gamma_s^{\frac{1}{2}} q_s^{\frac{1}{4}}
f^{-\frac{1}{4}}$. For the general triaxial ellipsoid this factor
depends on the intrinsic shape via the axial ratios $p_v$ and $q_v$
and the orientation $(\theta,\phi)$. For spheroids only a single axial
ratio and inclination angle are required.

We now consider the quantity,
\begin{equation}
Y_{\rm SZ} \equiv \frac{1}{D^2_A(z)}\frac{\sigma_T}{m_e c^2}
\int_V P_edV,
\end{equation}
where $V$ is the volume and $D_A(z)$ is the angular diameter distance
to the cluster. Since this quantity equals the integral of $y_c$ over
solid angle in the limit of a small angle subtended on the sky, it is
usually referred to as the ``integrated Compton-y
parameter''~\citep[e.g.,][]{whit02a}.  We will take $V$ to be a large
spherical or ellipsoidal region centered on the cluster.

\begin{thm} \label{thm.yz}
The deprojected spherical average of the integrated Compton-y parameter is,
\begin{equation}
\langle Y_{\rm SZ}(<r)\rangle^{\rm d+} =
\left(\gamma_s^2 q_s f^{\frac{1}{2}} p_vq_v\right)^{-1} Y_{\rm
SZ}(<a_v), \;{\rm where} \; a_v = (\gamma_s\sqrt{q_s})r. 
\end{equation}
\end{thm}

\begin{proof}
Starting from the definition of $Y_{\rm SZ}$, we have for a
spherical volume,
\begin{eqnarray}
\langle Y_{\rm SZ}(<r)\rangle^{\rm d+} & = & \frac{1}{D^2_A(z)}\frac{\sigma_T}{m_e c^2}
\int_0^r \langle P_e (r)\rangle^{\rm d}_{\rm SZ} 4\pi r^2dr\\
& = & \left(\gamma_s q_s^{\frac{1}{2}} f^{-\frac{1}{2}}\right)
\frac{1}{D^2_A(z)}\frac{\sigma_T}{m_e c^2} \int_0^r 
P_e(a_v) 4\pi r^2dr,
\end{eqnarray}
where we have used Corollary \ref{pressure.sz} and the fact that
$\gamma_s$, $q_s$, and $f$ depend only on $p_v$, $q_v$, and the fixed
line-of-sight projection orientation. Effecting a change of variable
within the integral from $r$ to $a_v = (\gamma_s\sqrt{q_s})r$ yields,
\begin{eqnarray}
\langle Y_{\rm SZ}(<r)\rangle^{\rm d+} & = & \left(\gamma_s^2 q_s
f^{\frac{1}{2}}\right)^{-1} \frac{1}{D^2_A(z)}\frac{\sigma_T}{m_e c^2}
\int_0^{a_v}  P_e(a_v) 4\pi a_v^2da_v\\
& = & \left(\gamma_s^2 q_s f^{\frac{1}{2}}\right)^{-1} \left[ \frac{Y_{\rm
SZ}(<a_v)}{p_vq_v}\right],
\end{eqnarray}
where the last step used the definition of an ellipsoidal volume
element, $dV = 4\pi p_vq_va_v^2da_v$, which proves the theorem. 
\end{proof}

Cosmological simulations predict a strong correlation between $Y_{\rm
SZ}$ and cluster mass~\citep[e.g.,][]{whit02a}, which is a direct
result of the gas pressure probing the depth of the cluster potential
well.  However, because simulations do not perfectly match observations
of cluster ICM~\citep[e.g., isophotal flattening of cool
cores,][]{fang09a}, and since it is preferable to use an independent
method to constrain cosmology, there is interest in using X-ray
observations of cluster mass to calibrate $Y_{\rm SZ}$ independently.
For high-quality X-ray data of clusters where it is possible to
measure accurately the spatially resolved gas density and temperature
profiles, direct calculation of the ICM pressure profile is to be
preferred for comparison to $Y_{\rm SZ}$~\citep{arna10a}. For lower
quality data it is necessary to rely on scaling relations, such as
$Y_{\rm X}$, a quantity advocated by~\citet{krav06a} as a mass proxy
for cosmological studies,
\begin{equation} \label{eqn.yx}
Y_{\rm X}(<r) = M_{\rm gas}(<r)T_{\rm X}(<r),
\end{equation}
where $M_{\rm gas}(<r)$ is the gas mass and $T_{\rm X}(<r)$ is the
emission-weighted temperature enclosed within the spherical volume of
radius $r$, and typically a radius $r_{500}$ is adopted. Since $Y_{\rm
X}$ is related to the integrated gas pressure profile, it should be
closely related to $Y_{\rm SZ}$. Indeed, recent SZ studies find a
strong correlation between $Y_{\rm SZ}(<r_{500})$ and $Y_{\rm
X}(<r_{500})$ and between $Y_{\rm SZ}(<r_{500})$ and $M_{\rm
gas}(<r_{500})$, each correlation having similar intrinsic
scatter~\citep[e.g.,][]{ande10a,plan11a} -- see also~\citet[][]{fabj11a}. 

Before addressing the spherical average of $Y_{\rm X}$, we consider
$T_{\rm X}$.

\begin{thm} \label{thm.tx}
The deprojected spherical average of the emission-weighted temperature
integrated over the spherical volume of radius $r$ is,
\begin{equation}
\langle T_{\rm X}(<r)\rangle^{\rm d+} = T_{\rm X}(<a_v), \;{\rm
where} \; a_v = (\gamma_s\sqrt{q_s})r.  
\end{equation}
\end{thm}

\begin{proof}
Beginning with the definition of $\langle T_{\rm X}(<r)\rangle^{\rm d+}$  
as the deprojected integrated emission-weighted temperature, we have,
\begin{equation}
\langle T_{\rm X}(<r)\rangle^{\rm d+}  = 
{ \int_0^r \langle (\epsilon T)(r)\rangle^{\rm d} 4\pi r^2dr \over 
\int_0^r \langle\epsilon (r)\rangle^{\rm d} 4\pi r^2dr} = {
\int_0^r \epsilon (a_v) T(a_v) 4\pi r^2dr \over  \int_0^r \epsilon (a_v) 4\pi r^2dr},
\end{equation}
where we made use of Theorem \ref{thm.general} in both the numerator
and denominator and the fact that the factor $(\gamma_s\sqrt{q_s/f})$
depends only on $p_v$, $q_v$, and the fixed line-of-sight projection
orientation. Effecting a change of variable within the integrals from
$r$ to $a_v = (\gamma_s\sqrt{q_s})r$ yields,
\begin{equation}
\langle T_{\rm X}(<r)\rangle^{\rm d+}  = 
{\int_0^{a_v} \epsilon (a_v) T(a_v) 4\pi a_v^2da_v \over  \int_0^{a_v}
\epsilon (a_v) 4\pi a_v^2da_v}.
\end{equation}
Multiplying the numerator and denominator of the R.H.S.\ by $p_vq_v$
yields the emission-weighted temperature integrated over the
volume of ellipsoidal radius $a_v = (\gamma_s\sqrt{q_s})r$, which is the desired result.
\end{proof}

We remark that Theorem~\ref{thm.tx} reduces to Theorem~\ref{thm.temp}
for the special case of a small radial volume element associated with
a finite radius $r$.

\begin{coro} \label{coro.yx}
The deprojected spherical average of $Y_{\rm X}(<r)$ is,
\begin{equation}
\langle Y_{\rm X}(<r)\rangle^{\rm d+}  = 
\left(\gamma_s^{\frac{5}{2}} q_s^{\frac{5}{4}} f^{\frac{1}{4}} p_vq_v\right)^{-1}
Y_{\rm X}(<a_v), \;{\rm where} \; a_v = (\gamma_s\sqrt{q_s})r.
\end{equation} 
\end{coro}

\begin{proof}
This result is an immediate consequence of the definition of $Y_{\rm
X}(<r)$ (eqn.\ \ref{eqn.yx}) and Theorems \ref{thm.mgas} and \ref{thm.tx}.
\end{proof}

\begin{coro} \label{coro.yszx}
The ratio of the deprojected spherical averages of $Y_{\rm SZ}$ and $Y_{\rm X}$ is,
\begin{equation}
\langle  Y_{\rm SZ,X}(<r)\rangle^{\rm d+}  \equiv \frac{\langle Y_{\rm
SZ}(<r)\rangle^{\rm d+}}{\langle Y_{\rm X}(<r)\rangle^{\rm
d+}} = \left(\gamma_s^{\frac{1}{2}} q_s^{\frac{1}{4}}
f^{-\frac{1}{4}}\right) \frac{Y_{\rm SZ}(<a_v)}{Y_{\rm X}(<a_v)} = 
 \left(\gamma_s^{\frac{1}{2}} q_s^{\frac{1}{4}}
f^{-\frac{1}{4}}\right)Y_{\rm SZ,X}(<a_v), \;{\rm where} \; a_v = (\gamma_s\sqrt{q_s})r.
\end{equation} 
\end{coro}

\begin{proof}
This result follows immediately from Theorem \ref{thm.yz} and
Corollary \ref{coro.yx}.
\end{proof}

\subsection{Connection to Stellar Dynamics}

Many of the results we have presented can be applied either directly,
or with minor modification, to a relaxed, dispersion-supported
collisionless stellar system with an isotropic velocity dispersion
tensor. Such a system obeys the equation of hydrostatic equilibrium
where the stellar density $\rho_{\rm stars}(a_v)$ replaces the ICM
density and the square of the velocity dispersion $\sigma(a_v)^2$
replaces the gas temperature. (Here it is assumed the stars, like the
gas, are merely a tracer of the gravitational potential, which is a
good approximation for galaxy clusters, and also for elliptical
galaxies well outside of the stellar half-light radius.) Consequently,
$\langle \rho_{\rm stars} (r)\rangle^{\rm d}$ obeys Theorem
\ref{thm.general}, and $\langle \sigma(<r)^2\rangle^{\rm d+}$
obeys Theorem
\ref{thm.temp}. Similarly, the deprojected spherically averaged mass
inferred from the stellar dynamics also obeys Theorem \ref{thm.m},
where in the proof one replaces the gas density with $\rho_{\rm
stars}(a_v)$ and the pressure with $\rho_{\rm
stars}(a_v)\sigma(a_v)^2$. Finally, the deprojected spherically
averaged stellar mass profile behaves as $\langle Y_{\rm
SZ}(<r)\rangle^{\rm d+}$ (Theorem \ref{thm.yz}), because the
deprojected stellar mass density behaves as $\langle P_e
(r)\rangle^{\rm d}_{\rm SZ}$ (Corollary \ref{pressure.sz}).

\section{Projection of Ellipsoidal Shells and Onion Peeling Deprojection}
\label{proj.ep}

To treat the case of binned observational data, such as the
one-dimensional surface brightness profile of a cluster, here we
describe the projection and deprojection of a system of concentric,
similar triaxial ellipsoidal shells relevant for the study of EPs.
For an ellipsoidal shell defined between $a_v^{\rm in}$ and $a_v^{\rm
  out}$ with constant emissivity, $\epsilon (a_v^{\rm in},a_v^{\rm
  out})$, throughout the shell, equation (\ref{eq.ep.proj}) becomes,
\begin{equation}
\Sigma (a_v^{\rm in}, a_v^{\rm out}; a_s) = \frac{2\epsilon
(a_v^{\rm in},a_v^{\rm out})}{\sqrt{f}}
\left[ \sqrt{(a_v^{\rm out})^2 - a_s^2} -  \sqrt{(a_v^{\rm in})^2 - a_s^2} \right],
\end{equation}
where $a_s \le a_v^{\rm out}$, and the second term in brackets is set
to zero if $a_s > a_v^{\rm in}$. This equation projects a
three-dimensional ellipsoidal shell onto a two-dimensional elliptical
surface brightness that depends only on the elliptical coordinate
$a_s$. We desire the luminosity integrated over an elliptical annulus
defined between semi-major axes, $\alpha^{\rm in}$ and $\alpha^{\rm
out}$:
\begin{eqnarray}
L(a_v^{\rm in}, a_v^{\rm out}; \alpha^{\rm in}, \alpha^{\rm out})& = &
\int_{\alpha^{\rm in}}^{\alpha^{\rm out}}\Sigma(a_v^{\rm in}, a_v^{\rm
out}; a_s) 2\pi q_s\alpha d\alpha\\
& = & \epsilon (a_v^{\rm in},a_v^{\rm out}) V^{\rm
int}(a_v^{\rm in}, a_v^{\rm out}; \alpha^{\rm in}, \alpha^{\rm out}), 
\end{eqnarray}
where,
\begin{eqnarray}
\lefteqn{V^{\rm int}(a_v^{\rm in}, a_v^{\rm out}; \alpha^{\rm in}, \alpha^{\rm
out}) = \frac{4\pi}{3}p_vq_v \, \times} \nonumber\\
& & \left( \left[ (a_v^{\rm out})^2 - (\gamma_s
q_s \alpha^{\rm in})^2\right]^{3/2} - \left[ (a_v^{\rm out})^2 - (\gamma_s
q_s \alpha^{\rm out})^2\right]^{3/2} +
\left[ (a_v^{\rm in})^2 - (\gamma_s q_s \alpha^{\rm out})^2\right]^{3/2}
- \left[ (a_v^{\rm in})^2 - (\gamma_s q_s \alpha^{\rm
    in})^2\right]^{3/2} \right). \label{eqn.vint}
\end{eqnarray}
If any terms in equation (\ref{eqn.vint}) have negative arguments,
they must be set to zero.  Note when viewed ``edge-on''
$(\theta=\phi=90^{\circ})$, i.e., down the intermediate principal
axis, ellipsoids always have $q_s=q_v$ and $\gamma_s q_s=1$. For the
special case of an oblate spheroid ($p_v=1$) viewed at arbitrary
inclination, $q_s=q_v\sqrt{f}$ and $\gamma_s q_s=1$.

Let an ellipsoid be partitioned into a series of concentric, similar
ellipsoidal shells, $a_{v,0} < a_{v,1} < a_{v,2} < \cdots < a_{v,{\rm
    N}}$. Define a corresponding set of concentric, similar,
elliptical annuli such that, $\alpha_0 < \alpha_1 < \alpha_2 < \cdots
< \alpha_{\rm N}$, where $\alpha_0 = a_{v,0}, \, \alpha_1 = a_{v,2},
\dots$. For this case, we may represent the projection of the
three-dimensional ellipsoidal shell $(a_{j-1},a_j)$ onto the
two-dimensional elliptical annulus $(\alpha_{i-1},\alpha_i)$ by,
\begin{equation}
V^{\rm int}_{ji}\equiv  V^{\rm int}(a_{v,{j-1}}, a_{v,j}; \alpha_{i-1}, \alpha_i) .
\end{equation}
That is, each shell and annulus is labeled by the index of its outer
boundary. The contribution of shell $j$ to the luminosity of annulus
$i$ is, $L_i = \epsilon_j V^{\rm int}_{ji}$, where
$\epsilon_j\equiv\epsilon (a_{v,{j-1}}, a_{v,j})$ is the constant
emissivity within the shell. We obtain the total luminosity projected
into annulus $i$ by summing the contributions from all shells $j\ge
i$,
\begin{equation}
L_i = \sum_{j=i}^{\rm N}\epsilon_j V^{\rm int}_{ji}, \label{eqn.epproj}
\end{equation}
so that the surface brightness is,
 \begin{equation}
\Sigma_{i} = \frac{L_i}{\pi q_s\left(\alpha_i^2 -
\alpha_{i-1}^2\right)}  = \frac{1}{\pi q_s\left(\alpha_i^2 -
\alpha_{i-1}^2\right)}\sum_{j=i}^{\rm N}\epsilon_j V^{\rm int}_{ji}. \label{eqn.surf.ep}
\end{equation}
Hence, we have shown that the projection matrix for spherical shells
(e.g., equation B12 of~\citealt{gast07b}) is generalized to the case
of ellipsoidal symmetry via the following mapping:
three-dimensional radius, $r\rightarrow a_v$; two-dimensional radius,
$R\rightarrow a_s = \gamma_s q_s \alpha$; and $4\pi/3 \rightarrow
p_vq_v4\pi/3$.  Moreover, by separating the first term from the others
in the summation of eqn.~(\ref{eqn.epproj}),
\begin{equation}
L_i = \epsilon_iV^{\rm int}_{ii} + \sum_{j=i+1}^{\rm N}\epsilon_j
V^{\rm int}_{ji},
\end{equation}
and then solving for the emissivity in shell $i$,
\begin{equation}
\epsilon_i = \left(\frac{L_i}{V^{\rm int}_{ii}}\right) - \sum_{j=i+1}^{\rm N}\epsilon_j
\left(\frac{V^{\rm int}_{ji}}{V^{\rm int}_{ii}}\right), \label{eqn.epdeproj}
\end{equation}
we arrive at a generalization of the ``onion peeling'' deprojection
method~\citep{deproj,kris83} appropriate for triaxial ellipsoids
$(p_v,q_v)$ with any orientation $(\theta,\phi)$.  That is, the
emissivity in shell $i$ is obtained by taking the total luminosity
observed in annulus $i$ on the sky, subtracting from it the luminosity
contributions projected from shells at larger radii ($j>i$), and then
dividing by $V^{\rm int}_{ii}$ representing the volume of intersection
between the ellipsoidal shell $i$ and the ellipsoidal cylindrical
shell defined by the base of elliptical annulus $i$ and an infinite
height. Practical implementation of eqn.\ (\ref{eqn.epdeproj})
requires assuming values for $\theta$, $\phi$, $p_v$, and $q_v$ and
measuring the value of $q_s$ on the sky.

While the above has focused on the example of the emissivity
projecting into the surface brightness, generalization to other
quantities such as the emission-weighted temperature and projected
temperature map ($\langle T\rangle_i$) is straightforward; i.e.,
\begin{equation}
\langle T\rangle_i= \frac{1}{L_i}\sum_{j=i}^{\rm N}\epsilon_j T_j V^{\rm int}_{ji}.
\end{equation}

\section{Conclusions}
\label{conc}

This is the first of two papers investigating the deprojection and
spherical averaging of ellipsoidal galaxy clusters (and massive
elliptical galaxies). We specifically consider applications to X-ray
and SZ studies, though many of the results also apply to isotropic
dispersion-supported stellar dynamical systems. A major disadvantage
of working with ellipsoidal systems, as opposed to spherical systems,
is that they generally involve numerical evaluation of computationally
expensive integrals. Here we present analytical formulas for galaxy
clusters described by a gravitational potential that is a triaxial
ellipsoid of constant shape and orientation; i.e., an ``ellipsoidal
potential'' (EP), which depends only on ellipsoidal radius, $\Phi =
\Phi(a_v)$.  While the mass density is itself not ellipsoidal for
these models, and it can take unphysical values in the vicinity of the
minor axis when the flattening is too large, we demonstrate that the
total mass enclosed within the ellipsoidal radius $a_v$ is
proportional to $d\Phi/da_v$, and is therefore well-behaved for any
smooth $\Phi(a_v)$, making it useful for many purposes.

We show that for hydrostatic X-ray studies of EPs the relationship
between the enclosed total mass, ICM temperature, and ICM density has
the same form (up to a proportionality factor) as the spherical case
where the spherical radius $r$ is replaced by $a_v$. Using this
result, along with the general result we derive for the spherical
deprojection of any ellipsoid of constant shape and orientation, we
show that the mass bias due to spherically averaging X-ray
observations is independent of the temperature profile. For the
special case of a scale-free logarithmic EP $(\Phi\propto \ln a_v)$
there is exactly zero bias for any shape, orientation, and temperature
profile. The ratio of spherically averaged ICM pressures obtained from
SZ and X-ray measurements depends only on the intrinsic shape and
projection orientation of the EP, as well as $H_0$, which provides
another illustration of how cluster geometry can be recovered through
a combination of X-ray and SZ measurements, with the key advantage
that the pressures are measured in the context of spherical symmetry
without (in principle) having to specify a parametric form for the
radial profile.  We also demonstrate that $Y_{\rm SZ}$ and $Y_{\rm X}$
have different biases as a result of spherical averaging, which lead
to an offset in the spherically averaged $Y_{\rm SZ} - Y_{\rm X}$
relation.  Paper~2 explores in more detail the biases and scatter
arising from spherical averaging, in particular using the
observationally and cosmologically motivated NFW mass profile, and
also considers the more widely investigated, and computationally
expensive, class of potentials where the mass density, rather than the
potential itself, is an ellipsoid of constant shape and orientation.

A potentially useful application of these analytical formulas is to
assess the error range on an observable accounting for deviations from
assumed spherical symmetry without having to perform the ellipsoidal
deprojection explicitly. That is, an X-ray observer can, as is
standard, analyze a cluster assuming spherical symmetry and obtain
deprojected ICM temperature and density profiles using the spherical
onion peeling procedure~\citep[see \S \ref{proj.ep};][]{deproj,kris83}
to construct, e.g., the observed spherically averaged mass profile,
$\langle M(<r)\rangle^{\rm d+}$. With the aid of Theorem \ref{thm.m}
this can be converted into the true mass profile $M(<a_v)$ for an
assumed three-dimensional shape and viewing orientation. Then by using
Theorem~\ref{thm.epspmass} the true spherically averaged mass profile
$\langle M(<r)\rangle^{\rm true}$ can be constructed. By adopting
priors for the shape and orientation and marginalizing over them, the
full range of $\langle M(<r)\rangle^{\rm true}$ owing to intrinsic
ellipsoidal geometry can be computed.

Finally, for dedicated ellipsoidal studies, we also generalize the
spherical onion peeling method to the triaxial case for a given shape
and orientation. The formulas presented for ellipsoidal shells may
also be of use for ellipsoidal projections in numerical work.

\section*{Acknowledgments}
We thank the anonymous referee for a timely and constructive
review. We gratefully acknowledge partial support from the National
Aeronautics and Space Administration under Grant No.\ NNX10AD07G
issued through the Astrophysics Data Analysis Program.

\bibliographystyle{mn2e}

\end{document}